\documentclass[preprint,authoryear,12pt]{elsarticle}

\usepackage{verbatim} 
\usepackage{amssymb}
\usepackage{amsmath}
\usepackage{amsthm}
\usepackage{epstopdf}
\usepackage{setspace}
\usepackage{rotating}
\usepackage[T1]{fontenc}
\usepackage{subfig}
\usepackage{algorithm2e}
\usepackage{chngcntr} 
\usepackage{lineno}
\usepackage{setspace} 
\usepackage[framed,numbered,autolinebreaks,useliterate]{mcode}
\newtheorem{definition}{Definition}[section]

\newtheorem{theorem}{Theorem}[section]

\numberwithin{equation}{section}
\usepackage{rotating}

\usepackage{epstopdf}
\usepackage{adjustbox}
\usepackage{blindtext}
\journal{arXiv.org}
\begin{document}
\begin{frontmatter}
\title{A novel algorithm to get the Fourier power spectra of a real sequence}
\author{Jiasong Wang$^{1}$, Changchuan Yin$^{2,\ast}$}
\address{1.Department of Mathematics, Nanjing University, Nanjing, Jiangsu 210093, China\\
2.Department of Mathematics, Statistics, and Computer Science, The University of Illinois at Chicago, Chicago, IL 60607-7045, USA\\
$\ast$ Corresponding author, email: cyin1@uic.edu
}
\begin{abstract}
For a real sequence of length of $m = nl,$ we may deduce its congruence derivative sequence with length of $l$. The discrete Fourier transform of original sequence can be calculated by the discrete Fourier transform of the congruence derivative sequence. Based on the relation of discrete Fourier transforms between the two sequences, the features of Fourier power spectra of the integer and fractional periods for a real sequence have been investigated. It has proved mathematically that after calculating the Fourier power spectrum at an integer period, the Fourier power spectra of the fractional periods associated this integer period can be easily represented by the computational result of the Fourier power spectrum at the integer period for the sequence.  A computational experience using a protein sequence shows that some of the computed results are a kind of Fourier power spectra corresponding to new frequencies which can${\text{'}}$t be obtained from the traditional discrete Fourier transform. Therefore, the algorithm would be a new realization method for discrete Fourier transform of the real sequence. 
\end{abstract}
\end{frontmatter}
\section{Introduction}
\label{sec1} 
Digital signal processing (DSP) techniques have been widely applied in the periodicity analysis of time series and becomes a major research technique for bioinformatics \citep{anastassiou2001genomic,chen2003will}. The main requisite in applying signal processing to symbolic sequences is to map the sequences onto numerical time series. For example, the Voss \citep{voss1992evolution} or Z-curve \citep{zhang1994z} representations may encode a symbolic DNA sequence as a numerical sequence, and hydrophobicity mapping may transform a protein sequence to numerical one \citep{kyte1982simple}. After numerical mapping, DSP methods can be employed to study features, structures, and functions of the symbolic sequences. The most common signal processing approach is Fourier transform \citep{welch1967use}, which has been commonly used to study periodicity and repetitive regions in symbolic DNA sequences \citep{silverman1986measure,anastassiou2001genomic}, as well as in genome comparison \citep{yin2014measure,yin2015improved}. We have surveyed the mathematical properties of the Fourier spectrum for symbolic sequences \citep{wang2014SNR}. For a review of the DSP methods for the study of biological sequences, one may refer to our book \citep{wang2013numerical}.

Periods in biological sequences can be categorized into two types: integer periods and fractional periods. For integer periods, the 3-base periodicity of protein-coding regions is often used in gene finding \citep{tiwari1997prediction,yin2005fourier,yin2007prediction,yin2014representation}; 2-base periodicity was found in introns of genomes \citep{arques1987periodicities,zhao2018detecting}. 2-periodicity exists in protein sequence regions for $\beta$ -sheet structures. Fractional periods, corresponding to fractional cycles in numerical series, are prevalent and important in structures and functions of protein sequences and genomes. For example, the 3.6-periodicity in protein sequences determines the $\alpha$-helix secondary structure \citep{EISENBERG1984,gruber2005repper,leonov2005periodicity,yin2017coevolution}. Strong 10.4- or 10.5-base periodicity in genomes are associated with nucleosomes \citep{trifonov19983,salih2015strong}. The 6.5-base periodicity is present in \textit{C. elegans} introns \citep{messaoudi2013detection}. Fractional period spectrum may offer high resolution and precise features of sequences, but accurately identifying fractional periods in a large genome is challenging. The demand of the periodicity research of DNA and protein sequences motivates us to investigate advanced techniques for fractional periodicity analysis of time series.

Due to the critical applications of Fourier transform for solving the problems of science and technology fields, new advancements of Fourier transform are constantly emerging. One of these achievements, fractional Fourier transform (FRFT), is a generalization of the Fourier transform, rediscovered many times over the past hundred years \citep{ozaktas1996digital,sejdic2011fractional}. FRFT can be realized as the development of continues Fourier transform. Discrete fractional Fourier transform (DFRFT) should be considered the extension of discrete Fourier transform (DFT) \citep{candan2000discrete}. FRFT and DFRFT have been successfully used to analyze the time$-$frequency information in quantum mechanics and quantum optics, the study of time$ -$frequency distributions, and many other applications \citep{sejdic2011fractional}. The emergence of various forms of Fourier transformation inspires people to find new formulas of Fourier transform.
It is known that Fourier power spectra of traditional Fourier transform, continuous or discrete fractional Fourier transform are often used as a primary criterion for their applications. For example, based on the power spectra of DFRFT of the DNA sequences, the phylogenetic trees of numerous species can be constructed for evolutionary history \citep{qian2018phylogenetic}. Based on our experience in calculating discrete Fourier power spectrum, a new realization method of DFT is proposed in this paper.

We previously proposed a method, periodic power spectrum (PPS), which directly computes Fourier power spectrum based on periodic distributions of signal strength on periodic positions \citep{profWang,yin2016periodic}. The advantage of the PPS method is that it avoids spectral leakage and reduces background noise, which both appear in the power spectrum of Fourier transform. Therefore, the PPS method can capture all latent integer periodicities in DNA sequences. We have utilized this method in the detection of latent periodicities in different genome elements, including exons and microsatellite DNA sequences. 

Based on the main idea of PPS method, to employ periodic distributions of the elements in a real sequence, efficiently and directly calculates Fourier power spectra at a specific integer period of the sequence. By using the computing results of Fourier power spectrum for the integer period an algorithm of easily calculating Fourier power spectra of the fractional periods associated the integer period for a real sequence is suggested and the theoretically mathematical proof of the algorithm is presented in this work. Numerical experience shows that it is an extension of DFT realization, similarly DFRFT. 

\section{Theorems and concepts}
\subsection{Fourier transform}
For a numerical sequence $x$ of length $m$, $x_0, x_1, \ldots, x_{m - 1}$, its discrete Fourier transform at frequency $k$ is defined as \citep{welch1967use}
\begin{equation}
X(k) = \sum\limits_{j = 0}^{m - 1} {x_j e^{ - i2\pi kj/m} } ,i = \sqrt { - 1,} {\kern 1pt} {\kern 1pt} {\kern 1pt} k = 0,1, \ldots ,m - 1,
\end{equation}
and its Fourier power spectrum ($FPS$) at frequency $k$ is
 \begin{equation}
\begin{gathered}
FPS(k) = X(k){}^*X(k) \hfill \\
= (\sum\limits_{j = 0}^{m - 1} {x_j e^{ - i2\pi jk/m} } )^* (\sum\limits_{n = 0}^{m - 1} {x_n e^{ - i2\pi nk/m} } ),k = 0,1, \ldots ,m - 1, \hfill \\ 
\end{gathered} 
\end{equation}
where $*$ indicates the complex conjugate.

\subsection{Concepts}
The Fourier power spectra of a signal can be represented by the congruence derivative sequence of the original sequence.

\begin{definition} For a real number sequence $x$ of length $m$, if two positive integers $n$ and $l$ satisfy $m = nl$, the congruence derivative sequence of $x$, $y$, length of $l$, and its element is defined by  \\
\begin{equation}
 	y_t  = \sum\limits_{j = 0}^{n - 1} {x_{jl + t} } ,t = 0,1, \ldots l - 1,
\end{equation}
and is named congruence derivative sequence or modulo distribution sequence. If the sequence length does not satisfy $m = nl$, then by padding zeros makes $m = nl$.
\end{definition}
Its DFT is 
 \begin{equation}
Y(r) = \sum\limits_{t = 0}^{l - 1} {y_t } e^{ - i2\pi tr/l} ,r = 0,1, \ldots ,l - 1.
\end{equation}
 
We introduce the relationship of DFTs between the original sequence and its congruence derivative sequence \citep{profWang}. 
\begin{theorem}
For a real sequence $x$ of length $m$, suppose $m = nl$, then the DFT of the sequence $x$ at frequency $kn$ is equal to the DFT of its congruence derivative sequence at frequency $k$.
 \begin{equation}
X(kn) = Y(k),0 \leqslant k \leqslant l - 1.
\end{equation}
\end{theorem}
  
\begin{proof}
  We know that \\
\[
\begin{gathered}
Y(k) = \sum\limits_{t = 0}^{l - 1} {y_t } e^{ - i2\pi tk/l}  = \sum\limits_{t = 0}^{l - 1} {\sum\limits_{j = 0}^{n - 1} {x_{jl + t} } e^{ - i2\pi tk/l} }  \hfill \\
= \sum\limits_{t = 0}^{l - 1} {\sum\limits_{j = 0}^{n - 1} {x_{jl + t} } e^{ - i2\pi (jl + t)k/l} }  = \sum\limits_{r = 0}^{m - 1} {x_r } e^{ - i2\pi rkn/m}  \hfill \\
= X(kn) = X(\frac{k}
{l}m). \hfill \\ 
\end{gathered} 
\]
\end{proof}
It is clear that $k = 0$ in formula (2.5) is the trivial case.

We know that $X(\frac{k}
{l}m)$ is the DFT of sequence $x$ at frequency $k/l$, in other word, $X(\frac{k}{l}m)$ is the DFT of sequence $x$ at period $l/k$, and the corresponding Fourier power spectrum of $x$ at periods $l/k$ is written as $FPS(l/k)$, where $k = 1, 2, \ldots, l - 1$. When $k = 1$, we call $FPS (l)$ as Fourier power spectrum at integer period $l$, and $FPS (l/k)$, $k = 2, 3,\ldots,l-1$ are named the associated fractional Fourier power spectra for the integer period $l$ because all the numerators of fractional periods are $l$.

We here introduce the self $q-shift$ summation of a sequence $y$.
\begin{definition}
	For a real sequence $y$ of length $l$, and $t = 0, 1, \ldots ,l - 1$, its self $q$-shift summation is defined by 
  \begin{equation} 
	z_{l,q}  = \sum\limits_{t = 0}^{l - 1} {y_t y_{t + q} } 
   \end{equation}
   with $t+q$ taken modulo $l$, $0\leqslant q < l$.
\end{definition}

A special case is for $q = l/2$ when $l$ is an even number, 
\begin{equation} 
z_{l/2,l/2}  = \sum\limits_{t = 0}^{l/2 - 1} {y_t y_{t + l/2} } 
 .
\end{equation} 

If we write sequence $y_t$, $t = 0,1, \ldots ,l - 1$, as a vector, $y^T = [y_0 ,y_1 , \cdots ,y_{l - 1} ]$, then $z_{l, q} $ can be realized by the autocorrelation of vector $y$. For a special case, $z_{l,0}$ is equal to the inner product of vector $y$, $z_{l,0}  = y^Ty$.

\subsection{Computations of Fourier power spectra of integer period and its associated fractional periods for a real sequence}

In formula (2.5) if $k = 1$, we may get the DFT of sequence $x$ at integer period $l$. Our previous research suggests a fast computing algorithm to compute Fourier power spectra at integer periods of a real sequence and its application to capture all latent integer periodicities in DNA sequences \citep{yin2016periodic}, and presented a fundamental algorithm of calculating $FPS$ of fractional periods for a real sequence \citep{wang2016fast}. 

Here, we further investigate the properties of integer and fractional periods of the real sequence. Noticed the formula (2.4) $Y(k) = \sum\limits_{t = 0}^{l - 1} {y_t } e^{ - i2\pi kt/l}$ is the DFT of real sequence, $y_{0}, y_1, \ldots , y_{l-1} $. When $k = 1$, $Y(l)$ is the DFT of the integer period $l$ for the sequence; when $k \ne 1$, and $Y(k)$ is the DFT of the fractional periods $l/2, l/3, \ldots, l/(l-1)$, respectively. These values $Y(k), k = 2, 3, \ldots, l-1$ are DFTs of the fractional periods that are associated with the integer period $l$.

\begin{theorem}
The Fourier power spectra of the integer period $l$ and its associated fractional periods for a real sequence are symmetric.
\end{theorem}

\begin{proof} 
Due to the conjugate symmetry of Fourier transform of the real sequence,
\[
Y(k) = Y(l - k)^* ,k = 1,2,3, \ldots ,l - 1.
\]
The power spectrum, $FPS(k) = Y(k)*Y(k)$, satisfies the following symmetric property
\begin{equation} 
FPS(k) = FPS(l - k),k = 1,2,3, \ldots ,l - 1.
\end{equation} 
\end{proof}
This theorem indicates that if we need computing the $FPS$ of the integer period $l$ and all its related fractional periods, calculating $ \left\lfloor {\frac{l}{2}}\right\rfloor$ $FPS$ is sufficient.

Let $y$ denote
\[
y = \left[ {y_{0,} y_{1,}  \ldots ,y_{l - 1} } \right]^T 
\]
and $e$ denote  
\[
e = \left[ {e^{ - \frac{{i2\pi k0}}
		{l}} ,e^{ - \frac{{i2\pi k1}}
		{l}} , \ldots ,e^{ - \frac{{i2\pi k(l - 1)}}
		{l}} } \right]^T ,
\]
where $T$ is the transpose of a vector. Formula (2.5) can be rewritten as 
\[
X(kn) = e^T y,
\]
therefore, Fourier power spectrum of the real sequence $x$ is as follows
\begin{equation} 
X(kn)^* X(kn) = Y(k)^* Y(k) = (e^t y)^* e^t y = y^t e^* e^t y.
\end{equation} 
The coefficient matrix of this quadratic form is written as
\[
B = e^* e^t .
\]
\begin{theorem}
Matrix $B = e^* e^t $ is a Hermitian Toeplitz matrix.
\end{theorem}
\begin{proof} 
Suppose the elements of $B$ are $b_{rs}$, $r, s = 1, 2, \ldots , l$, $b_{rs}$ can be represented as
\begin{equation} 
b_{rs}  = e^{i\frac{{2\pi kr}}
	{l}} e^{ - i\frac{{2\pi ks}}
	{l}}  = e^{i\frac{{2\pi k(r - s)}}
	{l}} .
\end{equation} 
If $r = s$, then  $b_{rs} = 1$, i.e. the diagonal elements of $B$ are ones, else if $r > s$, then $r-s > 0$, $ b_{rs}$ is at the lower triangular matrix of $B$, otherwise $b_{rs}$ is at the upper triangular matrix of $B$. Notice the representation of $b_{rs}$ in formula (2.10), $b_{rs}  = b_{sr} ^*$, therefore, matrix $B$ is a Hermitian matrix. 

A Toeplitz matrix is a matrix that has identical elements along any line parallel to diagonal. In matrix $B$, all elements along the lines parallel to diagonal their differences of the subscripts satisfy $r - s = 1$ or $ s - r = 1$, $r - s = 2$ or $s - r = 2$, $\ldots$, $r - s = l -1$ or $s - r = l-1 $, respectively, therefore, matrix $B$ is a Hermitian Toeplitz matrix.
\end{proof}

The formula (2.9) is the mathematical representation for $FPS(l/k)$ of a real sequence. From computational mathematics of view, for efficiently calculating $FPS(l/k)$, we need alternative procedure instead of directly computing (2.9) for saving computational cost.

\begin{theorem} For a real sequence $x$ of length $m$, if $m = nl$ and its congruence derivative sequence is $y_t, t = 0, 1, \cdots, l-1$, then $FPS(l/k), k = 1, 2, 3, \ldots, l - 1$ of the sequence can be expressed as follows.

If $l$ is an odd number,
\begin{equation} 
FPS(l/k) = z_{l,0}  + 2\sum\limits_{q = 1}^{\left\lfloor {\frac{l}
		{2}} \right\rfloor } {z_{l,q} \cos ((q)} \frac{{2k\pi }}
{l}).
\end{equation} 
  
Otherwise,
 \begin{equation} 
FPS(l/k) = z_{l,0}  + 2\sum\limits_{q = 1}^{\frac{l}
	{2} - 1} {z_{l,q} \cos ((q)\frac{{2k\pi }}
	{l}) + } 2z_{\frac{l}
	{2},\frac{l}
	{2}} \cos (\left\lfloor {\frac{l}
	{2}} \right\rfloor \frac{{2k\pi }}
{l}).
\end{equation} 
\end{theorem}

\begin{proof} 
Recall that the entries of matrix $B$ satisfy $b_{rs}  = b_{sr} ^*  = b_{lr - s} ^*$  and matrix $B$ is a Toeplitz matrix. Notice the variables of the quadratic form (2.9) with the property: $y_r y_t = y_t y_r$, and the definitions of  $z_{l,q} $ and $z_{l/2,l/2} $. So the expansions of (2.9), $FPS (l/k)$, can be decomposed into three parts. The first part is the diagonal elements of matrix $B$ corresponding to the quadratic expression of the congruence derivative sequence, $y_{t}, t = 0, 1, \ldots, l - 1$, that is,
\begin{equation}
z_{l,0}  = \sum\limits_{t = 0}^{l - 1} {y_t y_t } 
.
\end{equation}
The second part is the quadratic expression of the entries of lower triangular matrix of $B$. Without loss of generality, if $l$ is an odd number, we have the following.
\begin{equation}
\begin{gathered}
\sum\limits_{t = 0}^{l - 2} {y_{t + 1} y_t e^{\frac{{i2\pi k}}
		{l}} }  + y_{l - 1} y_0 e^{\frac{{i2\pi k(l - 1)}}
	{l}}  + \sum\limits_{t = 0}^{l - 3} {y_{t + 2} y_t e^{\frac{{i2\pi 2k}}
		{l}} }  + (y_{l - 2} y_0  + y_{l - 1} y_1 )e^{\frac{{i2\pi k(l - 2)}}
	{l}}  \hfill \\
+  \ldots  + \sum\limits_{t = 0}^{\frac{{l - 1}}
	{2}} {y_{t + \frac{{l - 1}}
		{2}} y_t e^{\frac{{i2\pi (\frac{{l - 1}}
				{2})k}}
		{l}} }  + \sum\limits_{t = 0}^{\frac{{l - 1}}
	{2} - 1} {y_{t + \frac{{l + 1}}
		{2}} y_t e^{\frac{{i2\pi (l - \frac{{l - 1}}
				{2})k}}
		{l}} } . \hfill \\ 
\end{gathered} 
\end{equation}
The third part is the quadratic expression of the entries of upper triangular matrix of $B$. According to conclusion of theorem 2.3 we may see that the third part is equal 
\begin{equation}
\begin{gathered}
\sum\limits_{t = 0}^{l - 2} {y_{t + 1} y_t e^{\frac{{ - i2\pi k}}
		{l}} }  + y_0 y_{l - 1} e^{\frac{{i2\pi k(1 - l)}}
	{l}}  + \sum\limits_{t = 0}^{l - 3} {y_t y_{t + 2} e^{\frac{{ - i2\pi 2k}}
		{l}} }  + (y_0 y_{l - 2}  + y_1 y_{l - 1} )e^{\frac{{i2\pi k(2 - l)}}
	{l}}  \hfill \\
+  \ldots  + \sum\limits_{t = 0}^{\frac{{l - 1}}
	{2}} {y_t y_{t + \frac{{l - 1}}
		{2}} e^{\frac{{ - i2\pi (\frac{{l - 1}}
				{2})k}}
		{l}} }  + \sum\limits_{t = 0}^{\frac{{l - 1}}
	{2} - 1} {y_t y_{t + \frac{{l + 1}}
		{2}} e^{\frac{{-i2\pi (l - \frac{{l - 1}}
				{2})k}}
		{l}} } . \hfill \\ 
\end{gathered} 
\end{equation}
The summation of formulas (2.13), (2.14), and (2.15) is the value
\[
\begin{gathered}
FPS(l/k) = z_{l,0}  + 2\cos (\frac{{2k\pi }}
{l})(y_0 y_1  + y_1 y_2  + y_2 y_3  + y_3 y_4  + y_4 y_5  +  \ldots  + y_0 y_{l - 1} ) \hfill \\
+ 2\cos ((2)\frac{{2k\pi }}
{l})(y_0 y_2  + y_1 y_3  + y_2 y_4  + y_3 y_5  + y_4 y_6  +  \ldots  + y_{l - 2} y_0  + y_{l - 1} y_1 ) \hfill \\
+ 2\cos ((3)\frac{{2k\pi }}
{l})(y_0 y_3  + y_1 y_4  + y_2 y_5  + y_3 y_6  + y_4 y_5  +  \ldots  + y_{l - 3} y_0  + y_{l - 2} y_1  + y_{l - 1} y_2 ) +  \cdots  \cdots  \hfill \\
+ 2\cos ((\frac{{l - 1}}
{2})\frac{{2k\pi }}
{l})(y_0 y_{(l - 1)/2}  + y_1 y_{(l + 1)/2}  + y_2 y_{(l + 3)/2}  \ldots  + y_{l - (l - 1)/2} y_0  +  \ldots  + y_{l - 1} y_{(l - 3)/2} ). \hfill \\ 
\end{gathered} 
\]

Using notation $z_{l,q }$ and write $(l - 1)/2 = \left\lfloor {l/2} \right\rfloor$, the above formula can be written as follows: 
when $l$ is an odd number,
\[
FPS(l/k) = z_{l,0}  + 2\sum\limits_{q = 1}^{\left\lfloor {\frac{l}
		{2}} \right\rfloor } {z_{l,q} \cos ((q)} \frac{{2k\pi }}
{l}),
\]
similarly, when $l$ is an even number,
\[
FPS(l/k) = z_{l,0}  + 2\sum\limits_{q = 1}^{\frac{l}
	{2} - 1} {z_{l,q} \cos ((q)\frac{{2k\pi }}
	{l}) + } 2z_{\frac{l}
	{2},\frac{l}
	{2}} \cos ((\frac{l}
{2})\frac{{2k\pi }}
{l}).
\]
\end{proof}
 
From the formulas (2.11 and 2.12), for $k = 1$, to compute $FPS (l/k)$, we only need to calculate $\left\lfloor {l/2} \right\rfloor +1$ quadratic expressions of $ y_0, y_1, \ldots, y_{l-1}$, when $l$ is an odd number, otherwise, $\left\lfloor {l/2} \right\rfloor+0.5$ quadratic expressions and their coefficients: $1, cos ((1)2\pi/l), cos ((2)2\pi/l),\ldots, cos(\left\lfloor {l/2} \right\rfloor 2\pi/l)$, respectively.

According to spectrum symmetry formula (2.8) to compute the Fourier power spectra of the integer period $l$ and all its associated fractional periods, i.e., calculating all $FPS (l/k), k = 1, 2, \ldots, (l-1)$, need to calculate $\left\lfloor {l/2} \right\rfloor$ Fourier power spectra. It will have much efficient way for computing them, as follows.  

\begin{theorem} The $\left\lfloor {l/2} \right\rfloor-1$ Fourier power spectra of its associated fractional periods of the real sequence can be deduced by the results in computational procedure of $FPS$ at integer period $l$, $FPS(l)$.
\end{theorem}

\begin{proof} 
We can prove that the computing results at $k = 1$ is sufficient for the computing $FPS(l/k)$ for $k \ne 1$. For a fixed $k$, $k \ne 1$, from the result just mentioned from formula (2.8), the range of $k$ should be in $2 \leqslant k \leqslant \left\lfloor {l/2} \right\rfloor $. Based on the computing results at $k = 1$, We have obtained that $z_{l,0} ,z_{l,1} , \ldots ,z_{l,\left\lfloor {l/2} \right\rfloor  - 1} $ and their corresponding coefficients,
\[
1,2\cos ((1)\frac{{2\pi }}
{l}),2\cos ((2)\frac{{2\pi }}
{l}), \ldots ,2\cos ((\left\lfloor {l/2} \right\rfloor  - 1)\frac{{2\pi }}
{l}),
\]
respectively, and for $l$ is an odd number $z_{l,\left\lfloor {l/2} \right\rfloor }$ corresponds to coefficient $2\cos (\left\lfloor {l/2} \right\rfloor \frac{{2\pi }}{l})$ as well as the coefficient of $z_{l/2,l/2}$ is $2\cos (\left\lfloor {l/2} \right\rfloor \frac{{2\pi }}{l})$ when $l$ is an even number.

If we can represent any coefficient of $FPS (l/k)$ as one of
\[
(1,\cos ((1)\frac{{2\pi }}
{l}),\cos ((2)\frac{{2\pi }}
{l}), \ldots ,\cos ((\left\lfloor {l/2} \right\rfloor  - 1)\frac{{2\pi }}
{l}),\cos ((\left\lfloor {l/2} \right\rfloor )\frac{{2\pi }}
{l}))
\]
when we calculate $FPS(l/k), 2 \leqslant k \leqslant \left\lfloor {l/2} \right\rfloor $, then we may complete the proof, no matter $l$ is odd number or even number.

For a fixed $k,2 \leqslant k \leqslant \left\lfloor {l/2} \right\rfloor $, to compute $FPS(l/k)$, we need to get the coefficients $1,\cos ((1)\frac{{k2\pi }}
{l}),\cos ((2)\frac{{k2\pi }}
{l}), \ldots ,\cos ((\left\lfloor {l/2} \right\rfloor  - 1)\frac{{k2\pi }}
{l}),\cos (\left\lfloor {l/2} \right\rfloor \frac{{k2\pi }}
{l})$, i.e., $\cos ((q)\frac{{k2\pi }}
{l}), q = 0,1,2, \ldots ,\left\lfloor {l/2} \right\rfloor$.

Suppose $p$ is the remainder $(kq,l)$, then $\cos ((q)\frac{{k2\pi }}
{l}) = \cos ((t)\frac{{2\pi }}
{l})$ if $p < \left\lfloor {l/2} \right\rfloor ,t = p$, otherwise, notice that $\cos ((r - s)\frac{{k2\pi }}
{l}) = \cos ((l - (r - s)\frac{{k2\pi }}
{l})$, set $t = l-p$. Then $\cos ((t)\frac{{2\pi }}{l})$ is one of the following
\[
\cos ((0)\frac{{2\pi }}
{l}),\cos ((1)\frac{{2\pi }}
{l}),\cos ((2)\frac{{2\pi }}
{l}) \ldots ,\cos ((\left\lfloor {l/2} \right\rfloor  - 1)\frac{{2\pi }}
{l}),\cos ((\left\lfloor {l/2} \right\rfloor )\frac{{2\pi }}
{l}).
\]

It proved that all the coefficients of $FPS(l/k)$ may one-by-one come from 
\[
\cos ((0)\frac{{2\pi }}
{l}),\cos ((1)\frac{{2\pi }}
{l}),\cos ((2)\frac{{2\pi }}
{l}) \ldots ,\cos ((\left\lfloor {l/2} \right\rfloor  - 1)\frac{{2\pi }}
{l}),\cos (\left\lfloor {l/2} \right\rfloor \frac{{2\pi }}
{l}).
\]
\end{proof}
\section{Algorithms and computation}
\subsection{Algorithms}
For computing $FPS(l/k)$, $1 \leqslant k \leqslant \left\lfloor {l/2} \right\rfloor$, three functions are designed. (1) function $ y(x,l)$ is to construct the congruence derivative sequence $y = [y_{0,} y_1 , \ldots ,y_{l - 1} ]^t $ of the original sequence $x$, (2) function $Z(y,q)$ is to yield the autocorrelation of vector $y$ for shift $q$ (formula (2.6)), and (3) function $Z(\frac{l}
{2},\frac{l}
{2})$, 
\[
{\text{Z}}_{{\text{l/2,l/2}}}  = \sum\limits_{t = 0}^{l/2 - 1} {y_t y_{t + l/2} } ,
\]
when $l$ is an even number.   

To speed up the computation, the algorithm shall use a built-in cosine series, $Cos(q,l)$, which is defined as 
\[
Cos(q,l) = \cos (q\frac{{2\pi }}
{l}),q = 0,1,2, \ldots ,\left\lfloor {\frac{l}
	{2}} \right\rfloor ,
\]
and the indices of columns are $l=2,3,4,5,\ldots$ .

Algorithm 1. To compute the $FPS$ at period $l$ for a real sequence, length of $m = nl$.

Input sequence $x$;

Call $y(x,l)$, obtained congruence derivative sequence y;

Set $FPS(l) := Z(0,y)$, then

for $t=1$ to ($\left\lfloor {l/2} \right\rfloor - 1$), do $FPS(l) := FPS(l) + 2Z(t,y) Cos(t,l)$;
 
Then if $l$ is an odd number do $FPS(l) := FPS(l) + 2Z(\left\lfloor {l/2} \right\rfloor,l) Cos(\left\lfloor {l/2} \right\rfloor,l)$; otherwise, $FPS(l) := FPS(l) + 2Cos(\left\lfloor {l/2} \right\rfloor,l)Z(l/2,l/2)$.  

Algorithm 2. After running algorithm 1, we have obtained the congruence derivative sequence $y$ of the sequence $x$ and $Z(q,l)$, $q =0,1, 2, \ldots, \left\lfloor {l/2} \right\rfloor$. To calculate the Fourier power spectra of the fractional periods $FPS(l/k)$, $2 \leqslant k \leqslant \left\lfloor {l/2} \right\rfloor -1$, and $Z(l/2,l]$ when $l$ is an odd number, otherwise, $Z(l/2,l/2]$. 

For a fixed $k$, set$ FPS(l/k) := Z(0,l)$,

For $q = 1$ to $q = \left\lfloor {l/2} \right\rfloor$ do $p = remainder (kq,l)$;

If $p \leqslant \left\lfloor {l/2} \right\rfloor $, then $t = p$ else $t=l - p$ ;

If $l$ is an odd number, then $FPS(l/k):= FPS(l/k) +2Cos (t,l) Z(q,l)$, until $q = \left\lfloor {l/2} \right\rfloor$;

Otherwise, if $l$ is an even number, if $q < \left\lfloor {l/2} \right\rfloor$, then $FPS(l/k) := FPS(l/k)+2Cos (t,l) Z(q,l)$, else $FPS(l/k) := FPS(l/k) +2Cos (t,l)Z(l/2,l/2)$.

\subsection{Computational experiments}
We chose protein sesquiterpene synthase with $\alpha$-helix structure (PDB:4GAX) for computational experiment \citep{li2013rational}. Using hydrophobicity representation, the symbolic protein sequence is mapped to a numerical sequence. Then, integer periods $l = 18$ and $l = 36$ in the numerical sequence are examined, respectively. The computational results is listed as follows. For  $l = 18$, $FPS(18/1) = 563.84$, $FPS(18/2) = 5267.3$, $FPS(18/3) = 3936.1$, $FPS(18/4) = 1116.7$, $FPS(18/5) = 21864$, $FPS(18/6) = 4381.2$, $FPS(18/7) = 939.8$, $FPS(18/8) = 2165$, and $FPS(18/9) = 1681$, maximum $FPS$ is 21864 which appears at period, $l/k = 18/5 = 3.6 $. It shows $\alpha$-helix secondary structure with period 3.6. 

For $l = 36$, $FPS(36/k), k = 1, 2, \ldots, 18$, are 9718.7, 563.84, 13134, 5267.3, 13318, 3936.1, 10628, 1116.7, 9272.5, 21864, 1190.8, 4381.2, 14623, 939.8, 7902.5, 2165, 616.43, and 1681, respectively. Also the maximum $FPS$ is 21864 at period $l/k = 36/10 = 3.6$.

It is easy to see, for example, these spectra: $FPS(18/4)$, $FPS(18/5)$, $FPS(18/7)$ and $FPS(18/8)$, can${\text{'}}$t be obtained from the traditional DFT. If chose $l$ is a prime number which less then and is not a factor of the original sequence length, the $FPS$ of the $l$ associated fractional periods can${\text{'}}$t be obtained from the traditional DFT. 

\begin{figure}[tbp]
	\centering
	\subfloat[]{\includegraphics[width=3.5in]{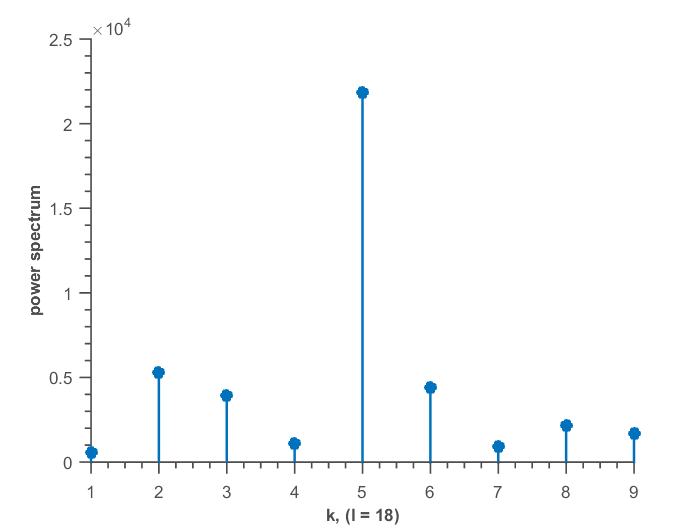}}\quad
	\subfloat[]{\includegraphics[width=3.5in]{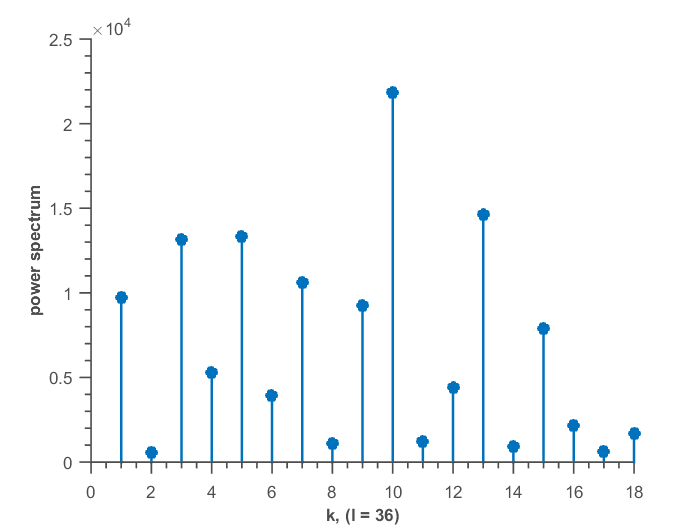}}\quad
	\caption{Fractional periodic spectrum analysis of protein sequence (PDB:4GAX). (a) $l = 18$, (b) $l = 36$.}
\end{figure}    

\section{Conclusion and prospect}
Over the years, we have used Fourier transform to study the structures, functions, interactions, and evolution of biological sequences, especially, for the periodicity's research of the sequences we suggested PPS method which instead of traditional DFT to help the studying of those sequences. Based on this investigation, in this paper, we suggest the algorithms of directly computing Fourier power spectra for an integer period and its associated fractional periods for a real sequence. The algorithm is out of the ordinary that after figure out the Fourier power spectra of the integer period then very easily and efficiently get the power spectra of its associated fractional periods of the sequence because the mathematical proof ensures the validity of the algorithm. 

The computational experiences show that it not only proves the efficient of the algorithm, but also suggests a new DFT computational realization method.

There are still rooms for improving the algorithm, we will continue to study the mathematical principles of the method. Their wide applications are more worthy of discoveries and researches in the future.

\newpage

\bibliographystyle{elsarticle-harv}
\bibliography{../References/myRefs}

\end{document}